\documentclass[submission,copyright,creativecommons]{eptcs}
\providecommand{\event}{QPL 2018} 
\usepackage{breakurl}             
\usepackage{underscore}           
\usepackage{amsmath}
\usepackage{wrapfig}
\usepackage{amssymb}
\usepackage{amsthm}
\usepackage{tikz-cd}
\usetikzlibrary{fit,decorations.pathreplacing,circuits.ee.IEC}  
\newtheorem{theorem}{Theorem}
\newtheorem{lemma}[theorem]{Lemma}
\newtheorem{proposition}[theorem]{Proposition}
\newtheorem{corollary}[theorem]{Corollary}
\theoremstyle{definition}
\newtheorem{definition}[theorem]{Definition}

\newcommand{\CC}{\mathbb C}

\newcommand{\defeq}{\stackrel{\mathrm{def}}=}
\newcommand{\id}{\mathrm{id}}

\newcommand{\swp}{\sigma}
\newcommand{\Id}{\mathrm{I}}

\newcommand{\AMCat}{\mathfrak{AMCat}}
\newcommand{\SMCat}{\mathfrak{SMCat}}

\newcommand{\mitem}[1]{\vspace{\itemsep}\par\noindent\setlength{\leftskip}{\leftmargin}\hspace{-\leftmargin}\textbf{#1}\hspace*{\labelsep}}
\newcommand{\mend}{\setlength{\leftskip}{0cm}}

\DeclareMathSymbol{\subfac}{\mathord}{operators}{"3C}

\title{Universal Properties in Quantum Theory}
\author{Mathieu Huot
\institute{ENS Paris-Saclay}
\and
Sam Staton
\institute{University of Oxford}
}
\def\titlerunning{Universal Properties in Quantum Theory}
\def\authorrunning{M.~Huot \& S.~Staton}
\begin{document}
\maketitle 

\begin{abstract}
We argue that notions in quantum theory should have universal properties in the sense of category theory. 
We consider the completely positive trace preserving (CPTP) maps, the basic notion of quantum channel. 
Physically, quantum channels are derived from pure quantum theory by allowing discarding. 
We phrase this in category theoretic terms by showing that the category of CPTP maps is the universal monoidal category with a terminal unit that has a functor from the category of isometries. In other words, the CPTP maps are the affine reflection of the isometries.
\end{abstract}

\section{Introduction}
\newcommand{\Isometry}{\mathbf{Isometry}}
\newcommand{\Injection}{\mathbf{Injection}}
\newcommand{\Tennent}{\mathbf{Tennent}}
\newcommand{\Function}{\mathbf{Function}}
\newcommand{\Embedding}{E}
\newcommand{\CPTP}{\mathbf{CPTP}}
\newcommand{\CatA}{\mathbf{C}}
\newcommand{\CatB}{\mathbf{D}}
\newcommand{\tr}{\mathrm{tr}}
\newcommand{\QIsometry}{\mathbf{Isometry_2}}
\newcommand{\QCPTP}{\mathbf{CPTP_2}}
\newcommand{\Inj}{\mathit{Inj}}
The basic foundation of statistical quantum mechanics and quantum channels is usually motivated as follows.
\begin{description}
\item[Step 1.] Pure quantum theory is not random, and is moreover reversible. 
\item[Step 2.] Pure quantum theory does not allow us to discard or hide parts of a system.
\item[Step 3.] Full quantum theory accounts for the perspective of an observer for whom some things are hidden.  Hiding/discarding parts of a system can lead to randomness, mixed states, and quantum channels. 
\end{description}
\par
\newcommand{\ket}[1]{\ensuremath{\left|#1\right\rangle}} 
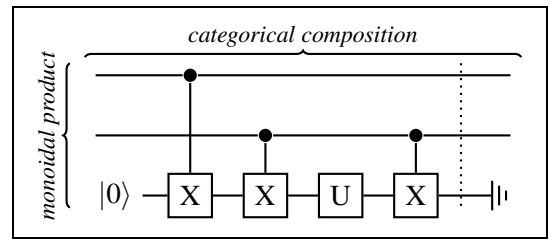
\begin{wrapfigure}[9]{r}{7cm}%
\vspace{-12pt}
{\framebox{    \begin{tikzpicture}[thick,circuit ee IEC,yscale=0.8]
    \tikzstyle{operator} = [draw,fill=white,minimum size=1.5em] 
    \tikzstyle{phase} = [fill,shape=circle,minimum size=5pt,inner sep=0pt]
    %
    \node at (-0.4,0) (q1) {};
    \node at (-0.4,-1) (q2) {};
    \node at (0,-2) (q3) {\ket{0}};
    %
    \node[phase] (op11) at (1,0) {} edge [-] (q1);
    \node[operator] (op13) at (1,-2) {X} edge [-] (q3);
    \draw[-] (op11) -- (op13);
    %
    \node[phase] (op22) at (2,-1) {} edge [-] (q2);
    \node[operator] (op23) at (2,-2) {X} edge [-] (op13);
    \draw[-] (op22) -- (op23);
    %
    \node[operator] (op33) at (3,-2) {U} edge [-] (op23);
    %
    \node[phase] (op42) at (4,-1) {} edge [-] (op22);
    \node[operator] (op43) at (4,-2) {X} edge [-] (op33);
    \draw[-] (op42) -- (op43);
    %
    \node (o1)  at (5.4,0) {} edge [-] (op11);
    \node (o2)  at (5.4,-1) {} edge [-] (op42);
    \node[ground,right] (o3)  at (5,-2) {} edge [-] (op43);
    %
    \draw[decorate,decoration={brace},thick] (-0.6,-2.2) to
	node[midway,left] (bracket) {\rotatebox{90}{\footnotesize \emph{monoidal product}}}
	(-0.6,0.2);
    \draw[decorate,decoration={brace},thick] (-0.4,0.3) to
	node[midway,above] (bracket) {\footnotesize \emph{categorical composition}}
	(5.4,0.3);
    \draw[dotted] (4.6,0.2) to (4.6,-2.2);
    \end{tikzpicture}
  }}
  \caption{\label{figure}
 \emph{A typical quantum circuit.}
  }
\end{wrapfigure}

In this paper we propose to formalize this argument in categorical terms as follows.
We use the language of (symmetric) monoidal categories, which are structures that support two forms of 
combination, as illustrated in Figure~\ref{figure}: the monoidal product for juxtaposing systems, and categorical composition for connecting 
the inputs/outputs of systems. The figure also illustrates the discarding of an ancilla (notated \begin{tikzpicture}[circuit ee IEC,yscale=0.6,xscale=0.5]
\draw (0,0) to (2ex,0) node[ground,xshift=.65ex] {};
\end{tikzpicture}). 

\mitem{Step 1.} Pure quantum theory is based on the monoidal category $\Isometry$ of finite dimensional Hilbert spaces and isometries between them. Recall that an isometry $\CC\to \CC^n$ is an n-level pure state, and every isometry $\CC^n\to\CC^n$ is invertible (unitary).

\mitem{Step 2.} A monoidal category admits discarding when its monoidal unit (representing the empty system) is a terminal object. For then every system $A\otimes B$ has a canonical map $A\otimes B\to A\otimes 1\cong A$, discarding $B$. 
But in the category of isometries, the monoidal unit $\CC$ is not a terminal object. 

\mitem{Step 3.}
Full quantum theory can be interpreted in any symmetric monoidal category 
that contains $\Isometry$ 
but where the unit is  
a terminal object (it supports discarding).
\emph{Our main theorem} (Theorem~\ref{theorem}) is that  the universal such category is
the monoidal category $\CPTP$ of finite dimensional Hilbert spaces and completely positive trace preserving maps between them. Recall that a CPTP map $\CC\to \CC^n$ is an $n$-level mixed state in the usual sense. Thus full quantum theory is canonically determined from pure quantum theory by the universal property of having a terminal unit. 

\mend

\section{Preliminaries}
For completeness, we recall some basic ideas from category theory (e.g.~\cite{joyalstreet}) 
as well as the definitions of the key categories under consideration (e.g.~\cite[Ch.~4]{wilde}, \cite{ck-pictures}). 
The expert reader can skip this section.
\begin{definition}
A \emph{strict monoidal category} 
is a category $\CatA$ together with a functor $\otimes:\CatA\times \CatA\to \CatA$ and an object $I$ such that 
$(A\otimes B)\otimes C=A\otimes (B\otimes C)$, $A\otimes I=I\otimes A=A$, 
$(f\otimes g)\otimes h=f\otimes (g\otimes h)$, and $f\otimes \id_I=\id_I\otimes f=f$. 
A \emph{symmetric strict monoidal category} additionally has a natural isomorphism
$\sigma_{A,B}:A\otimes B\cong B\otimes A$ such that $\sigma_{B,A}\sigma_{A,B}=\id_{A\otimes B}$ and 
\[
\begin{tikzcd}
A\otimes B\otimes C
\ar[d,"\sigma_{A,B}\otimes C"']\ar[dr,"\sigma_{A,B\otimes C}"]
\\B\otimes A\otimes C\ar[r,"B\otimes \sigma_{A,C}"']&B\otimes C\otimes A
\end{tikzcd}
\]
\end{definition}
In what follows, we only consider strict monoidal categories, and so we drop the adjective ``strict''. 
(This is primarily to avoid complicating the presentation with 2-categorical considerations.)
\begin{definition}[Example: Isometries]\label{isometry}
The symmetric monoidal category $\Isometry$ of isometries is defined as follows. 
The objects are natural numbers, with $n$ to be thought of as the Hilbert space $\CC^n$. The
morphisms $m\to n$ are linear maps $f\colon \CC^m\to \CC^n$ that are isometries $(||f(v)||=||v||)$. 
Categorical composition is just composition of linear maps. 
The monoidal structure is given on objects by multiplication: $m\otimes n\defeq mn$.
Note that the Hilbert space $\CC^{mn}$ is the tensor product of $\CC^m$ and $\CC^n$, 
and so we can define the monoidal structure on morphisms by taking the usual tensor product of linear maps. 

More concretely, every isometry is represented by a matrix $V\in \CC^{mn}$ such that $V^*V=\Id$.
Categorical composition is by matrix multiplication. 
The monoidal structure on morphisms is the Kronecker product.
\end{definition}
For example, a pure qubit state is the same thing as an isometry $1\to 2$, and the 
circuit in Figure~\ref{figure} to the left of the dotted line is an isometry $2\otimes 2\to 2\otimes 2\otimes 2$ (i.e.~$4\to 8$). 
\begin{definition}[Example: CPTP maps]
\label{def:cptp}\label{cptp}
The symmetric monoidal category $\CPTP$ of completely positive trace preserving maps is defined as follows. 
The objects are natural numbers, with $n$ to be thought of as the space $M_n=\CC^{n^2}$ of $n\times n$ complex matrices,
that is, as the algebra of operators on the $n$-dimensional Hilbert space.
The morphisms $m\to n$ are linear maps $M_m\to M_n$ that are completely positive and trace preserving. 
The idea is that these are the maps that preserve density matrices even when coupled with an ancilla. 

To formulate a formal definition, first note that $M_{mn}$ is the tensor product of Hilbert spaces $M_m$ and $M_n$,
and so we can combine any linear maps $f:M_m\to M_n$, $g:M_{m'}\to M_{n'}$ into a linear map 
$f\otimes g:M_{mm'}\to M_{nn'}$. 
A linear map $f:M_m\to M_n$ is \emph{positive} if $f(\rho)$ is a positive operator whenever $\rho$ is a positive operator
($v^\top\rho v\geq 0$ for all $v\in \CC^m$). 
The map $f:M_m\to M_n$ is \emph{completely positive} if for every natural number $p$ the linear map 
$(f\otimes p):M_{mp}\to M_{np}$ is positive. 
Categorical composition in $\CPTP$ is composition of linear maps. 
The monoidal structure is again given by multiplication $(m\otimes n\defeq mn)$ on objects and 
on morphisms by the tensor product of linear maps. 

In the literature, CPTP maps are sometimes called `quantum channels' or `superoperators'.
\end{definition}
A special aspect of $\CPTP$, compared to $\Isometry$, is that the monoidal unit $1$ is a terminal object. This means that there is a unique morphism ${!}:n\to 1$ for every object $n$.
This unique CPTP map ${!}:n\to 1$ is the trace operator $\tr(\rho)=\sum_{i=1}^n\rho_{i,i}$. 
The tensor $\id_m\otimes !:m\otimes n\to m\otimes 1=m$ is the partial trace operator
$M_{mn}\to M_m$. 
For example, the circuit in Figure~\ref{figure} to the right of the dotted line is the partial trace operator
$2\otimes 2\otimes 2\to 2\otimes 2$ (i.e.~$8\to 4$). 

\begin{definition}[Symmetric monoidal functor]
If $\CatA$ and $\CatB$ are symmetric monoidal categories, then a strict symmetric monoidal functor is a functor 
$F:\CatA\to \CatB$ such that $F(I_\CatA)=I_\CatB$, $F(A\otimes_\CatA B)=(F(A))\otimes_{\CatB} (F(B))$, and similarly with morphisms, and such that $F(\sigma_\CatA)=\sigma_\CatB$. 

If $V:m\to n$ is an isometry then the mapping $\rho\mapsto V\rho V^*$ is a completely positive trace preserving map $m\to n$. 
This extends to an identity on objects functor $E:\Isometry \to \CPTP$. 
\end{definition}
For example, the entire circuit diagram in Figure~\ref{figure} is a CPTP map $2\otimes 2\to 2\otimes 2$ (i.e.~$4\to 4$),
formed by composing the left hand part (an isometry $4\to 8$ considered as a CPTP map via $E$) with 
the right hand part (a partial trace operator $8\to 4$).

We remark that $E:\Isometry\to \CPTP$ is not faithful because the isometries $1\to 1$ are the global phase shifts, of which there are many, 
whereas the object $1$ is terminal in $\CPTP$. 
This is the collapse of global phase. 
\section{Main theorem}
In this section we state and prove our main theorem: that $\CPTP$ (Def.~\ref{cptp}) is the universal monoidal category with 
a terminal unit with a functor from $\Isometry$ (Def.~\ref{isometry}). 
\begin{theorem}\label{theorem}
For every strict symmetric monoidal category with terminal unit $\CatB$ 
and every symmetric monoidal functor $F:\Isometry\to \CatB$ 
there is a unique symmetric monoidal functor $\hat F:\CPTP\to \CatB$ making the following diagram commute. 
\begin{equation}\begin{tikzcd}
\Isometry \arrow[rd,"F"'] \arrow[r,"\Embedding"]&\CPTP\ar[d,dashed,"\hat F"]\\
&\CatB
\end{tikzcd}
\label{eqn:theorem}
\end{equation}
\end{theorem}
Before we prove the theorem, we remark that it uniquely determines $\CPTP$ to within unique
isomorphism: if $E':\Isometry\to \CatA$ also has the unique extension property ($\forall F\exists! \hat F.F=\hat FE'$), then there is a unique 
symmetric monoidal isomorphism $J:\CPTP\cong \CatA$ making $E'=JE$.
This uniqueness is usual for a universal property (e.g.~\cite[\S III.1]{maclane}).
So this universal property could be used as a \emph{definition} of the category $\CPTP$. 
Indeed, we argue that this universal property is more directly motivated than Definition~\ref{def:cptp}, since it comes immediately from a categorical formulation of Steps 1--3 in the introduction. By contrast, to motivate Definition~\ref{def:cptp} directly, one must somehow explain and motivate mixed states and density matrices, then explain why a CPTP map is defined on all matrices not just density matrices, explain why complete positivity rather than positivity is required, and so on. (On the other hand we do not dispute the power of Definition~\ref{def:cptp} as a calculational tool.)

Theorem~\ref{theorem} is a consequence of Stinespring's dilation theorem and its uniqueness condition. We now recall the following 
variant of it.
See for instance, \cite[Thms.~2.2 \& 2.5]{wolf} or 
\cite[Def.~5.2.1, Ex.~5.2.5]{wilde}.
For diagrammatic argument, see \cite[Ch.~6]{ck-pictures};
for an alternative universal uniqueness condition and generalizations, see \cite[Prop.~13]{ww-paschke};
for an analysis from the perspective of probabilistic-theories, see~\cite{purification}.
\begin{lemma}\label{lemma}
Every completely positive trace-preserving map 
$f:m\to n$ can be written as 
\[
f(\varrho)=\tr_a(V\varrho V^*)
\]
for an isometry $V:m\to n\otimes a$, where $\tr_a=n\otimes !:(n\otimes a)\to n$ is partial trace. 

The choice of $V$ and $a$ is unique up to isometries of the ancilla~$a$, in the following sense. 
If $W:m\to n\otimes b$ is an isometry and 
\[f=\tr_a(V(-)V^*)=\tr_b(W(- )W^*)\]
then there are isometries 
$a\xrightarrow {V'} c\xleftarrow{W'}b$ 
such that $(\id_n\otimes V')\circ V = (\id_n\otimes W')\circ W$.
\end{lemma}

We now return to the main theorem.
\begin{proof}[Proof of Theorem~\ref{theorem}]
Since $\Embedding$ is identity on objects, we must define $\hat F(n)\defeq F(n)$.

For a CPTP map $f:m\to n$ we use Lemma~\ref{lemma} to pick a Stinespring dilation $(V,a)$ and define $\hat F(f):F(m)\to F(n)$ by
\[
\hat F(f)\defeq \Big(F(m)\xrightarrow{F(V)} F(n\otimes a)=F(n)\otimes F(a)\xrightarrow{F(n)\otimes !} F(n)\Big)
\]
This morphism $\hat F(f)$ is independent of the choice of dilation $(V,a)$. For if $(W,b)$ is another dilation, the uniqueness property in Lemma~\ref{lemma} guarantees that we can find $c$ such that the following diagram commutes:
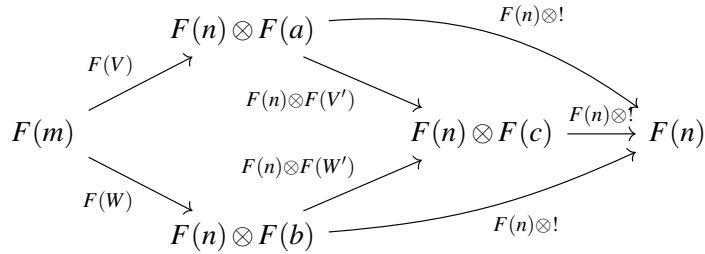
\[
\begin{tikzcd}
&F(n)\otimes F(a)\ar[dr,"F(n)\otimes F(V')"']\ar[drr,"F(n)\otimes !",bend left=20]
\\
F(m)\ar[ur,"F(V)"]\ar[dr,"F(W)"']&&F(n)\otimes F(c)\ar[r,"F(n)\otimes !"]&F(n)
\\
&F(n)\otimes F(b)\ar[ur,"F(n)\otimes F(W')"]\ar[urr,"F(n)\otimes !"',bend right=10]
\end{tikzcd}
\]
Notice that if any symmetric monoidal functor $\hat F$ is going to make diagram~\eqref{eqn:theorem}
commute then it must be defined in exactly this way because it must preserve the monoidal structure and composition, 
because $!$ is unique, and because it must commute with $F$. Thus, provided $\hat F$ is a symmetric monoidal functor,
it is unique. 

We must check that this definition preserves the monoidal category structure. 
Preservation of identity morphisms is easy: we can pick the dilation $(\Id,1)$. 
For preservation of composition, notice that if $(V,a)$ is a dilation of $f:m\to n$ and $(W,b)$ is a dilation of $g:n\to p$ 
then $((W\otimes \id_a)V,b\otimes a)$ is a dilation of $gf$. 
Indeed 
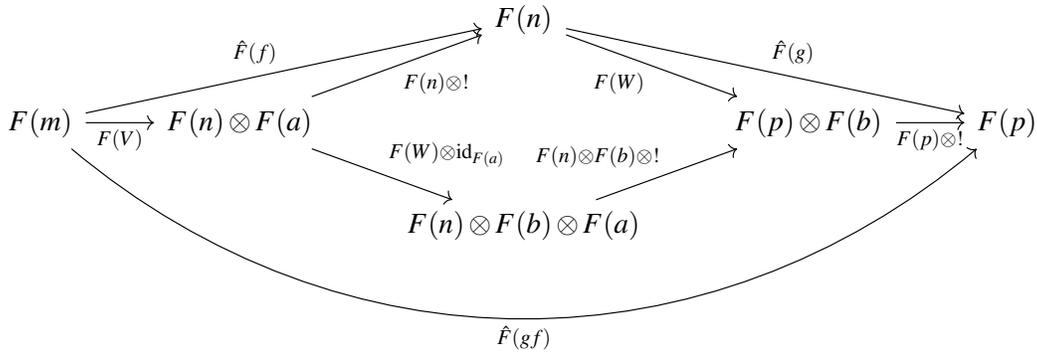
\[
\begin{tikzcd}
&&F(n)\ar[dr,"F(W)"']\ar[drr,"\hat F(g)"]&&
\\
F(m)\ar[urr,"\hat F(f)"]\ar[r,"F(V)"']\ar[rrrr,bend right=40,"\hat F(gf)"']
&F(n)\otimes F(a)\ar[ur,"F(n)\otimes !"']\ar[dr,"F(W)\otimes \id_{F(a)}"]
&&F(p)\otimes F(b)\ar[r,"F(p)\otimes !"']&F(p)
\\
&&F(n)\otimes F(b)\otimes F(a)\ar[ur,"F(n)\otimes F(b)\otimes !"]&&
\end{tikzcd}
\]
For preservation of the monoidal product, notice that if $(V,a)$ is a dilation of $f:m\to n$ and $(W,b)$ is a dilation of $g:p\to q$ 
then $((\id_m\otimes \swp\otimes \id_p)\circ (V\otimes W),a\otimes b)$ is a dilation of $f\otimes g$. 
Indeed 
\[
\begin{tikzpicture}[baseline= (a).base]
\node[scale=.95] (a) at (0,0){
\begin{tikzcd}
F(m)\otimes F(p)\ar[rrr,"\hat F(f)\otimes \hat F(g)",bend left=10]
\ar[r,"F(V)\otimes F(W)"']
\ar[rrr,"\hat F(f\otimes g)",bend right=40]
&
F(n)\otimes F(a)\otimes F(q)\otimes F(b)
\ar[rr,"F(n)\otimes !\otimes F(q)\otimes !"]
\ar[dr,"F(n)\otimes \swp\otimes F(b)"']
&&F(n)\otimes F(q)
\\
&&F(n)\otimes F(q)\otimes F(a)\otimes F(b)
\ar[ur,"F(n)\otimes F(q)\otimes !"]
\end{tikzcd}
};
\end{tikzpicture}
\vspace{-1cm}
\]
\end{proof}

\section{Discussion and outlook}

\paragraph{Terminal units and discarding in quantum theory and elsewhere.}
Monoidal categories with terminal unit are often called `affine monoidal categories' or `semi-cartesian monoidal categories'. The importance of this structure is widely recognized in both Categorical Logic and in the Categorical Quantum Mechanics programme. In categorical logic, affine monoidal categories have long been considered as a version of linear logic which is resource sensitive but where resources can be discarded (generally, see e.g.~\cite{jacobs-affine,walker-substructural}; in the quantum setting, see~\cite{sv-qlambda,adams-qpel}). In CQM, traces have long been regarded as important in the Selinger's CPM construction (e.g.~\cite{selinger-cpm,coecke-cp}, and \cite[Ch.~6]{ck-pictures}, but also \cite{ch-ax-cp,ch-purity,ktw-qkd,ku-causal}). Effectus theory is a bridge between logic and CQM, and affine monoidal categories play a key role there
(e.g.~\cite[\S 10]{cjww-effectus}, \cite{aj-comet}, \cite{tull-otp}). 
Despite all this interest, Theorem~\ref{theorem} appears to be novel.

\paragraph{Affine reflections in general.}
We can restate our main theorem in terms of an adjunction (following e.g.~\cite[Thm.~IV.1.2(ii)]{maclane}).
The category $\SMCat$ of (small) symmetric monoidal categories and symmetric monoidal functors has a full subcategory 
$\AMCat$ comprising those monoidal categories for which the unit is terminal. 
The full and faithful embedding $\AMCat\to \SMCat$ has a left adjoint 
$L:\SMCat\to\AMCat$.
In other words,
$\AMCat$ is a reflective subcategory of $\SMCat$.
The universal property of $\CPTP$ from Theorem~\ref{theorem} can be rephrased in these terms as follows.
\begin{corollary}The symmetric monoidal category of CPTP maps is the affine reflection of the symmetric monoidal category of isometries:
\[L(\Isometry)\cong \CPTP\] \label{corollary}
\end{corollary}
The reflection $L$ has also been investigated by Hermida and Tennent~\cite[Cor.~2.11]{hermida-tennent}. They work in the dual, `co-affine', setting, and they use the construction for the different purpose of modelling specification logics for non-quantum programs. 
Dualizing their analysis, we see that in general the category $L(\CatA)$ can be described as having the same objects as $\CatA$ but 
the morphisms $m\to n$ are equivalence classes of pairs ($a$, $f\colon m\to n\otimes a$) for the equivalence relation generated by 
\begin{equation}\label{hermida-tennent}
(a,f)\sim (b,(n\otimes g)f)\quad \text{for }f\colon m\to n\otimes a,\ g\colon a\to b
\end{equation}
In Appendix~\ref{appendix} we discuss the affine reflection of the category of injections; see also~\cite[\S 4]{hermida-tennent} where other examples are also given. 

\paragraph{Further comparison with the CPM construction.}
The general relationship between the affine reflection $L$ and the CPM construction from categorical quantum mechanics warrants further investigation. 
We make some preliminary remarks. First we recall that the CPM construction has been described as an initial object in a category~\cite{ch-ax-cp,coecke-cp}, but it is unclear how this initiality
relates abstractly to our universal property. 
However, we can compare the CPM approach with our approach based on affine reflections as follows. Recall that
the CPM construction considers those maps between abstract matrix algebras for which there exists a dilation. 
Thus 
\begin{itemize}
\item Two dilations are equated in the CPM construction when they give rise to the same map between abstract matrix algebras;
\item Two dilations are equated in the affine reflection when they are equal according to the equivalence relation generated by~\eqref{hermida-tennent}. 
\end{itemize}
These two notions of equivalence of dilation coincide when dilations are essentially unique, as in Lemma~\ref{lemma}, but this is unlikely to be the case in an arbitrary categorical setting.
Indeed, this essential uniqueness of dilations is often taken as a postulate for reconstructing quantum mechanics (e.g.~\cite{purification}).

\paragraph{Relation to quantum circuits}

When $\CatA$ is a PROP, \emph{i.e.} there is an object $A$ such that every object is of the form $A\otimes \dots \otimes A$, then $L(\CatA)$ is again a PROP, which is obviously presented (in the sense of App. $A.2$ of \cite{baez2017props}, Sec. $2.1$ of \cite{bonchi2017interacting}) by one generating morphism $\begin{tikzpicture}[circuit ee IEC,yscale=0.6,xscale=0.5]
\draw (0,0) to (2ex,0) node[ground,xshift=.65ex] {};
\end{tikzpicture}:1\to 0$ and equations of the form:
\begin{equation}
\begin{tikzpicture}[circuit ee IEC,yscale=0.6,xscale=0.5,baseline={([yshift=-.5ex]current bounding box.center)}]
    \tikzstyle{operator} = [draw,fill=white] 
\node (op0) at (-1,0) {};
\node[operator,xshift=.65ex,minimum size=4em] (op1) at (-2,1) {f};
\draw (-0.2,1.5) to (2ex,1.5) node[ground,xshift=.65ex] {};
\draw (-0.2,0.5) to (2ex,0.5) node[ground,xshift=.65ex] {};
\draw[-] (-3.3,0) to (-4,0);
\draw[-] (-3.3,1) to (-4,1);
\draw[-] (-3.3,2) to (-4,2);
\end{tikzpicture} 
\quad = \quad
\begin{tikzpicture}[circuit ee IEC,yscale=0.6,xscale=0.5,baseline={([yshift=-.5ex]current bounding box.center)}]
\draw (0,2) to (2ex,2) node[ground,xshift=.65ex] {};
\draw (0,1) to (2ex,1) node[ground,xshift=.65ex] {};
\draw (0,0) to (2ex,0) node[ground,xshift=.65ex] {};
\end{tikzpicture}
\qquad\qquad\qquad\qquad
\begin{tikzpicture}[circuit ee IEC,yscale=0.6,xscale=0.5,baseline={([yshift=-.5ex]current bounding box.center)}]
    \tikzstyle{operator} = [draw,fill=white] 
\node (op0) at (-1,0) {};
\node[operator,xshift=.65ex,minimum size=4em] (op1) at (-2,1) {g};
\draw (-0.2,2) to (2ex,2) node[ground,xshift=.65ex] {};
\draw (-0.2,1) to (2ex,1) node[ground,xshift=.65ex] {};
\draw (-0.2,0) to (2ex,0) node[ground,xshift=.65ex] {};
\draw[-] (-3.3,0.5) to (-4,0.5);
\draw[-] (-3.3,1.5) to (-4,1.5);
\end{tikzpicture} 
\quad = \quad
\begin{tikzpicture}[circuit ee IEC,yscale=0.6,xscale=0.5,baseline={([yshift=-.5ex]current bounding box.center)}]
\draw (0,1.5) to (2ex,1.5) node[ground,xshift=.65ex] {};
\draw (0,0.5) to (2ex,0.5) node[ground,xshift=.65ex] {};
\end{tikzpicture}
\end{equation}
and so on.

In particular, if we focus on the PROP of quantum circuits without discarding and measurement, $\QIsometry$, which is the full subcategory of $\Isometry$ generated by $2$, 
then we can form $L(\QIsometry)$ by freely adding discarding. In  Appendix~\ref{appendix2} we show that $L(\QIsometry)$ is a full subcategory of $\CPTP$, by using a variation of the proof of Theorem \ref{theorem}. We understand that there is ongoing work to add discarding to the ZX-calculus \cite{jeandel2018complete}, which may be along similar lines. 

\paragraph{Summary and directions.}
The main idea of this paper is that the categories that arise in quantum theory can and should be made canonical by virtue of universal properties that are motivated by physics. The universal property in Theorem~\ref{theorem} and
Corollary~\ref{corollary} is directly motivated by the three steps of the physical argument in the introduction. 

Our starting point for this work was the programming-like axiomatization of completely positive unital maps between C*-algebras~\cite{staton-popl15}. This is now part of a bigger ongoing programme to investigate how to give universal properties to other aspects of operator algebras. For example:
\begin{itemize}
\item The category $\Isometry$ also has another important monoidal structure $\oplus$, 
given on objects by addition ($m\oplus n=m+n$), 
and this forms a bimonoidal category~\cite{laplaza}. We have preliminary results extending Theorem~\ref{theorem} 
to take account of this bimonoidal structure. It might be interesting to relate this to the CP* construction~\cite{cpstar}.
\item The isometries $m\to n$ have a topological structure, making $\Isometry$ a topologically enriched category. We have preliminary results extending Theorem~\ref{theorem} to give a canonical topological enrichment for $\CPTP$, using~\cite{cts-stinespring}. 
\end{itemize}
An ultimate goal is to investigate whether universal properties suggest new categories for quantum theory.

\paragraph{Acknowledgements.}
It has been helpful to discuss this work with Chris Heunen, Jamie Vicary, Bas Westerbaan, and the Nijmegen and Oxford quantum groups. Thanks to anonymous reviewers for their suggestions. 
The work was supported by ERC Grant QCLS, EPSRC Grant  EP/N007387/1, and a Royal Society University Research Fellowship.

\bibliographystyle{eptcsini}
\bibliography{refs}

\appendix
\section{Matrix algebra representation of the affine reflection of injections}
\label{appendix}
\newcommand{\Comparison}{F}
The main result of this paper (Thm.~\ref{theorem}) is that the affine reflection of the isometries is the category of CPTP maps. 
But one can construct an affine reflection of any symmetric monoidal category~\cite{hermida-tennent}. 
We now consider the case of injections, which might be thought of as a non-quantum analogue of pure reversibility.
\begin{definition}
The symmetric monoidal category $\Injection$ has objects natural numbers, considered as sets, and morphisms injections between them. 
Composition is composition of functions. The monoidal structure is multiplication on objects, and the pairing of injections on morphisms. 
\end{definition}
This category $\Injection$ can be thought of as a wide subcategory of the category of isometries, once we understand 
an injection $f\colon m\to n$ as an isometry
$V_f\colon m\to n$ which is the linear map $\CC^m\to \CC^n$ with $(V_f(\vec a))_{f(i)}=a_i$ and $(V_f(\vec a))_j=0$ if $j\not\in\mathsf{image}(f)$. 
 
The category $\Function$ of all functions between natural numbers is also a symmetric monoidal category with a terminal unit. 
But it is not the free affine reflection. 
\begin{proposition}
There is no symmetric monoidal functor $F\colon \Function\to\CPTP$ making the following diagram commute: 
\[
\begin{tikzcd}
\Injection \arrow[d,"V"'] \arrow[r]&\Function\ar[d,dashed,"F"]\\
\Isometry\arrow[r,"\Embedding"]&\CPTP
\end{tikzcd}\]
\end{proposition}
\begin{proof}
Suppose (to get a contradiction) that $F\colon \Function\to \CPTP$ is such a monoidal functor.
Consider the injection $f:3\to 6$ defined by $(1\mapsto 1,2\mapsto2,3\mapsto 6)$. Notice that $(!_2\otimes \id_3)\circ f=\id_3$ as functions. 
Since $F$ is functorial and monoidal, we have $F((!_2\otimes \id_3)\circ f)=(!_2\otimes \id_3)\circ F(f) =(!_2\otimes\id_3)\circ \Embedding(V_f)$. So for a $3\times 3$ matrix $M=(m_{i,j})_{1\leq i,j\leq 3}$ 
we must have \[F((!_2\otimes \id_3)\circ f)(M)=
\tr_2(V_fMV_f^*)=
\Big(\begin{smallmatrix} m_{1,1} & m_{1,2} & 0 \\ m_{2,1} & m_{2,2} &0 \\ 0&0&m_{3,3}\end{smallmatrix}\Big)\neq M\text.\] 
So $F((!_2\otimes \id_3)\circ f)\neq F(\id_3)$, a contradiction. 
\end{proof}

In his work on (non-quantum) specification logic, Tennent~\cite{tennent} proposed the following category 
in place of $\Function$. (Actually, Tennent also allowed infinite sets, and considered the dual category, but we skip over that for now.)
\begin{definition}[\cite{tennent}, \S3]
The category $\Tennent$ has objects natural numbers, and morphisms $m\to n$ are pairs $(Q,f)$ where $Q$ is an equivalence relation on $m$ and $f\colon m\to n$ is a function
that is injective on each $Q$-equivalence class: if $f(i)=f(i')$ and $Q(i,i')$ then $i=i'$. 
The identity morphism $m\to m$ is $(m\times m,\mathrm{id}_m)$, where $(m\times m)$ is the universal equivalence relation. The composite
$(R,g)\cdot (Q,f)$ is $(S,gf)$, where 
$S(i,i')$ iff $Q(i,i')$ and $R(f(i),f(i'))$. 
\end{definition}
Roughly, Tennent's intention was that the objects $m$ and $n$ are sets of allowed memory states,
and the morphisms describe how different memory states relate to each other. 

\begin{proposition}[Hermida and Tennent,~\cite{hermida-tennent}, Thm.~4.4]\label{prop:hermida}
The functor $\Injection\to \Tennent$, taking an injection $f\colon m\to n$ to $(m\times m,f)$, exhibits $\Tennent$ as the affine reflection of $\Injection$. 
\end{proposition}
As a corollary of Proposition~\ref{prop:hermida}, since $\CPTP$ has a terminal unit, there is an identity-on-objects symmetric monoidal functor $\Comparison\colon\Tennent\to \CPTP$ making the following diagram commute:
\[
\begin{tikzcd}
\Injection \arrow[d,"V"'] \arrow[r]&\Tennent\ar[d,dashed,"\Comparison"]\\
\Isometry\arrow[r,"\Embedding"]&\CPTP
\end{tikzcd}\]
One can use the techniques in~\cite{hermida-tennent} to calculate that the functor $\Comparison$ takes a morphism $(Q,f)\colon m\to n$ to the composite CPTP map:
\begin{equation}\label{comparison}
\Comparison(Q,f)\ =\ m\xrightarrow{V_{(f,q)}(-)V_{(f,q)}^*} n\otimes (m/\!_Q)\xrightarrow{\tr_{(m/\!_Q)}} n
\end{equation}
where $q\colon m\to m/\!_Q$ is the quotient of $m$ by the equivalence relation $Q$. Notice that although $f\colon m\to n$ need not be injective, the pair $(f,q)\colon m\to n\otimes (m/\!_Q)$ will always be injective \cite[Prop.~4.1]{hermida-tennent}. 
In fact, nothing is lost by regarding Tennent's morphisms as CPTP maps:
\begin{proposition}
The functor $\Comparison\colon \Tennent\to \CPTP$ is faithful. 
\end{proposition}
\begin{proof}
For $i,i'\leq m$, let $e_{i,i'}$ be the basic $m\times m$ matrix with $0$ everywhere except $1$ in $(i,i')$. By expanding \eqref{comparison} we see that 
\begin{equation}\label{equalfunctions}
((F(Q,f))(e_{i,i}))_{j,j}=\begin{cases}1&\text{if }j=f(i)
\\0&\text{otherwise}\end{cases}
\end{equation}
and 
\begin{equation}\label{equalrelations}
((F(Q,f))(e_{i,i'}))_{f(i),f(i')}=\begin{cases}1&\text{if }Q(i,i')
\\0&\text{otherwise.}\end{cases}
\end{equation}
Thus if $F(Q,f)=F(R,g)$ then $f=g$ by \eqref{equalfunctions} and $Q=R$ by \eqref{equalrelations}. 
\end{proof}
So we can equivalently understand Tennent's morphisms as those CPTP maps between matrix algebras for which there is an injective dilation. 

From the perspective of Lemma~\ref{lemma}, a curious corollary is that 
in this non-quantum setting, there is always a canonical choice of dilation for which the ancilla is smaller than $m$. 

\section{Specialisation to quantum circuits}
\label{appendix2}
The categories involved in the main result of the paper \ref{theorem} have objects natural numbers. Each natural number $n$ represents a state space of dimension $n$. In quantum computing the setting is often restricted to systems composed of qubits only. A qubit has a state space of dimension $2$ and formalisms using quantum circuits only consider systems of dimension $2^n$ for some natural number $n$, representing the evolution of $n$ qubits. 

We now consider a variation of Theorem \ref{theorem} where objects $n$ of our categories represent qubit state spaces of dimension $2^n$.

\begin{definition}
The symmetric monoidal category $\QIsometry$ has objects natural numbers and morphisms $m\to n$ are isometries $f:\mathbb{C}^{2^m}\to\mathbb{C}^{2^n}$. The monoidal structure is given on objects by addition (NB $2^{n}2^m=2^{n+m}$) and on morphisms by tensor product of linear maps. Thus it is a full monoidal subcategory of $\Isometry$.
\end{definition}

\begin{definition}
The symmetric monoidal category $\QCPTP$ has objects natural numbers and morphisms $m\to n$ are CPTP maps $f:\mathcal{M}_{2^m}\to \mathcal{M}_{2^n}$. The monoidal structure is given similarly on objects by addition and on morphisms by tensor product. Thus it is a full monoidal subcategory of $\CPTP$.
\end{definition}

The functor $\Embedding:\Isometry\to\CPTP$ given by $E(V)=\rho\mapsto V\rho V^*$ restricts to an identity-on-objects strict symmetric monoidal functor $\Embedding_2:\QIsometry\to\QCPTP$.

\begin{proposition}
For every symmetric strict monoidal category with terminal unit $\CatB$ and every symmetric strict monoidal functor $F:\QIsometry\to\CatB$ there is a unique symmetric strict monoidal functor ${\widehat{F}:\QCPTP\to\CatB}$ making the following diagram commute:
\begin{equation}\begin{tikzcd}
\QIsometry \arrow[rd,"F"'] \arrow[r,"\Embedding_2"]&\QCPTP\ar[d,dashed,"\hat F"]\\
&\CatB
\end{tikzcd}
\end{equation}
\end{proposition}
\begin{proof}
Uniqueness is obtained in the same way as in Theorem \ref{theorem}, using Lemma \ref{lemma}. 

For existence, the proof of Theorem \ref{theorem} needs to be modified because the dilation might not be a power of $2$.
In fact a similar argument goes through but we need to verify:
\begin{enumerate}
\item the existence of a dilation of a size a power of $2$
\item the independence of the choice of the dilation
\item that the defined $\widehat{F}$ preserves the symmetric monoidal structure
\end{enumerate}
Item $3$ is shown in the same way as in the proof of Theorem \ref{theorem}. We now elaborate Items $1$ and $2$.

Item 1. We show the existence of a dilation of the right size as follows.
Let $f:n\to m$ be a morphism in $\QCPTP$. 
By Lemma \ref{lemma} there is an isometry $U:2^n\to (2^ma)$ in $\Isometry$ such that $(U,a)$ is a Stinespring dilation for $f$. 

If $b\geq a$, let denote by $\Inj_{a,b}:a\to b$ the canonical injection represented by the matrix with $a$ 1 on the diagonal and $0$ elsewhere.
Let $V:(2^ma)\to (2^m2^k)$ be the isometry $id_{2^m}\otimes \Inj_{a,2^k}$ where $k$ is a natural number such that $2^k\geq a$. 
Then $(VU,2^k)$ is a Stinespring dilation for $f$. Let $\widehat{F}(f)=(Id_{F(2^m)}\otimes!)\circ F(VU)$.

Item 2. We show that $\widehat{F}(f)$ is independent of the choice of dilation. Let $f:m\to n$ be a morphism in $\QCPTP$ and $(V,2^k)$ and $(U,2^{k'})$ be two Stinespring dilations for $f$. By Lemma \ref{lemma} there are isometries $V',U'$ such that the following diagram commutes:

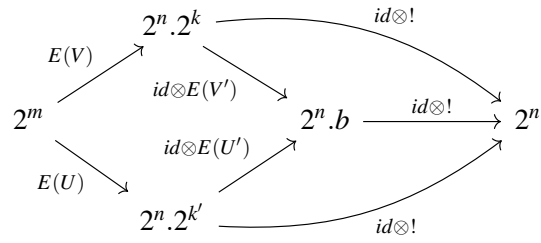
\[
\begin{tikzcd}
&&&
\\
& 2^n.2^k \ar[dr,"id\otimes E(V')"'] \ar[drrr,"id\otimes !",bend left=22]  &&  &
\\
2^m \ar[ur,"E(V)"] \ar[dr,"E(U)"']  && 2^n.b  \ar[rr,"id\otimes !"] && 2^n
\\
& 2^n.2^{k'} \ar[ur,"id\otimes E(U')"] \ar[urrr,"id\otimes !"',bend right=20]
\end{tikzcd}
\]

Using a similar map $id\otimes \Inj_{b,2^{b'}}$ as in the proof of point $1$, we obtain the following commuting diagram:

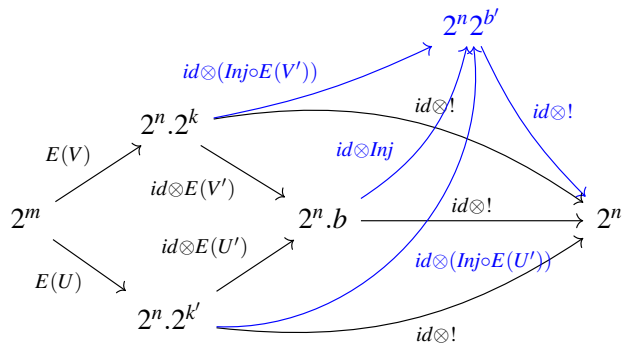
\[
\begin{tikzcd}
&&&  \textcolor{blue}{2^n2^{b'}} \ar[ddr, blue, "id\otimes !", bend right=10]
\\
& 2^n.2^k \ar[dr,"id\otimes E(V')"'] \ar[drrr,"id\otimes !",bend left=22] \ar[urr,blue,bend right= 7,"id\otimes (\Inj\circ E(V'))"] &&  &
\\
2^m \ar[ur,"E(V)"] \ar[dr,"E(U)"']  && 2^n.b \ar[uur,blue, "id\otimes \Inj", near start, bend right= 22] \ar[rr,"id\otimes !"] && 2^n
\\
& 2^n.2^{k'} \ar[ur,"id\otimes E(U')"] \ar[urrr,"id\otimes !"',bend right=20] \ar[uuurr,blue,bend right= 50,"id\otimes (\Inj \circ E(U'))"']
\end{tikzcd}
\]
and we conclude similarly as in the proof of Theorem \ref{theorem} that $(id\otimes !)\circ F(V)=(id\otimes !)\circ F(U)$.
\end{proof}
\end{document}